\documentclass[runningheads, envcountsame, a4paper]{llncs}%
\usepackage{etex}
\usepackage[pst,vaucanson-g]{xcolor}
\usepackage{vaucanson-g}
\usepackage{pstricks}
\usepackage{pst-tree,pst-node}
\usepackage{amsfonts}
\usepackage{tikz}
\usepackage{array}
  \newcommand{\defproblem}[3]{
  \vspace{2mm}
\noindent\fbox{
  \begin{minipage}{0.96\textwidth}
  #1\\
  {\bf{Input:}} #2  \\
  {\bf{Output:}} #3
  \end{minipage}
  }
  \vspace{2mm}
}

\usepackage[algo2e, noend, noline, boxed]{algorithm2e}
  \SetKwComment{tcm}{\{}{\}}
   \newenvironment{myfunction}[2][htbp]
  {%
    \setlength{\algomargin}{.2cm}
    \begin{center}
    \begin{minipage}{#2}
    \begin{function}[#1]
    \small
     \let\Par=\par
       \def\par{\endgraf\vspace{.1cm}}
           \SetKw{To}{to}%
       \SetKw{Downto}{downto}%
           \SetKw{Or}{or}%
       \SetKwFor{Algo}{Function}{}{}%
      \vspace{.15cm}%
   }
   {%
     \let\par=\Par
     \end{function}%
     \end{minipage}%
     \end{center}%
   }

\newcommand{\ACSM}{\textsc{ApproximateCircularStringMatching}}
\def\dd{\mathinner{.\,.}}
\newcommand{\cO}{\mathcal{O}}
\begin{document}
\frontmatter          % for the preliminaries
%
%\pagestyle{headings}  % switches on printing of running heads
%\addtocmark{Hamiltonian Mechanics} % additional mark in the TOC
%
\title{Average-case Optimal Approximate Circular String Matching}

\author{Carl Barton\inst{1}
\and Costas S.\ Iliopoulos\inst{1,2}
\and Solon P.\ Pissis\inst{1}\thanks{Supported by a London Mathematical Society grant (no. 51303).}
}

\institute{$\!^1\ $Department of Informatics, King's College London, The Strand, London, UK \\ 
\email{\{carl.barton,costas.iliopoulos,solon.pissis\}@kcl.ac.uk} \\ 
$\!^2\ $Department of Mathematics \& Statistics,
University of Western Australia, 35 Stirling Highway, Perth, Australia \\ }

\titlerunning{Average-Case Optimal Approximate Circular String Matching}
\authorrunning{C. Barton, C. S. Iliopoulos, and S. P. Pissis}

\toctitle{Average-Case Optimal Approximate Circular String Matching}
\tocauthor{Carl~Barton, Solon~P.~Pissis and Costas~S.~Iliopoulos}
\maketitle

\begin{abstract}
Approximate string matching is the problem of finding all factors of a text $t$ of length $n$
that are at a distance at most $k$ from a pattern $x$ of length $m$. 
Approximate circular string matching is the problem of finding all factors of $t$
that are at a distance at most $k$ from $x$ {\em or} from any of its rotations. 
In this article, we present a new algorithm for approximate circular string matching under the edit distance model with optimal average-case search time $\cO(n(k + \log m) /m)$.
Optimal average-case search time can also be achieved by the algorithms for multiple approximate string matching (Fredriksson and Navarro, 2004) using
$x$ and its rotations as the set of multiple patterns. Here we reduce the preprocessing time and space requirements compared to that approach.
\end{abstract}

\section{Introduction}
In order to provide an overview of our results and algorithms,
we begin with a few definitions, generally following~\cite{CHL07}.
We think of a \textit{string} $x$ of \textit{length} $n$ as an array
$x[0\dd n-1]$, where every $x[i]$, $0 \le i < n$, is a \textit{letter}
drawn from some fixed \textit{alphabet} $\Sigma$ of size $\sigma = \cO(1)$.
By a {\em $q$-gram} we refer to any string $x \in \Sigma^q$.
The \textit{empty string} of length $0$ is denoted by $\varepsilon$.
A string $x$ is a \textit{factor} of a string $y$ if there exist two strings $u$ and $v$, such that $y=uxv$.
Consider the strings $x,y,u$, and $v$, such that $y=uxv$. If $u=\varepsilon$, 
then $x$ is a \textit{prefix} of $y$. If $v=\varepsilon$, then $x$ is a \textit{suffix} of $y$.
Let $x$ be a non-empty string of length $n$ and $y$ be a string. 
We say that there exists an \textit{occurrence} of $x$ in $y$, or, more simply, that $x$
\textit{occurs in} $y$, when $x$ is a factor of $y$.
Every occurrence of $x$ can be characterised by a position in $y$. Thus we say that $x$ occurs at the
\textit{starting position} $i$ in $y$ when $y[i \dd i + n - 1]=x$.
Given a string $x$ of length $m$ and a string $y$ of length $n \geq m$, the \emph{edit distance}, 
denoted by $\delta_E(x,y)$, is defined as the minimum total cost of operations 
required to transform one string into the other. For simplicity, we only count the number of edit operations, 
considering the cost of each to be $1$~\cite{levelshtein-66-binary}.
The allowed edit operations are as follows:
\begin{itemize}
	\item \emph{Insertion}: insert a letter in $y$, not present in $x$; $(\varepsilon, b),~b \neq \varepsilon$
	\item \emph{Deletion}: delete a letter in $y$, present in $x$; $(a, \varepsilon),~a \neq \varepsilon$ 
	\item \emph{Substitution}: replace a letter in $y$ with a letter in $x$; $(a,b),~a \neq b,\texttt{and}~a,b \neq \varepsilon$. 
\end{itemize}

\noindent We write $x \equiv_k^E y$ if the edit distance between $x$
and $y$ is at most $k$. Equivalently, if $x \equiv_k^E y$, we say that
$x$ and $y$ have at most $k$ {\em differences}. 
We refer to the \textit{standard dynamic programming matrix} of $x$ and $y$ as the matrix defined by

\noindent$\textsf{D}[i,0] = i,\textrm{ }0\leq i\leq m,\textrm{ }\textsf{D}[0,j] = j,\textrm{ }0\leq j\leq n$
\[ \textsf{D}[i,j] =\min \left\{
  \begin{array}{l}
    \textsf{D}[i-1,j-1]  + (1 \textsf{ if } x[i-1]\neq y[j-1])\\
\textsf{D}[i-1,j] + 1 \\
\textsf{D}[i,j-1] + 1 
  \end{array} \right. , 1\leq i\leq m, 1\leq j\leq n.\]

\noindent Similarly we refer to the \textit{standard dynamic programming algorithm} as the algorithm to compute the 
edit distance between $x$ and $y$ through the above recurrence in time $\cO(mn)$. 
Given a non-negative integer threshold $k$ for the edit distance, this can be computed in time $\cO(mk)$~\cite{editd}.
We say that there exists an \textit{occurrence} of $x$ in $y$ with at most $k$ differences, or, more simply, that $x$
\textit{occurs in} $y$ with at most $k$ differences, when $u \equiv_k^E x$ and $u$ is a factor of $y$.

A circular string of length $n$ can be viewed as a traditional linear string which has the left- and right-most symbols 
wrapped around and stuck together in some way. Under this notion, the same circular string can be seen as $n$ different 
linear strings, which would all be considered equivalent. Given a string $x$ of length $n$, we denote 
by $x^{i}=x[i \dd n-1]x[0 \dd i-1]$, $0 < i < n$, the $i$-th \textit{rotation} of $x$ and $x^{0}=x$.
Consider, for instance, the string $x=x^{0}=\texttt{abababbc}$; this string has the following rotations:
$x^{1}=\texttt{bababbca}$, $x^{2}=\texttt{ababbcab}$, $x^{3}=\texttt{babbcaba}$, $x^{4}=\texttt{abbcabab}$, 
$x^{5}=\texttt{bbcababa}$, $x^{6}=\texttt{bcababab}$, $x^{7}=\texttt{cabababb}$.

This type of structure occurs in the DNA of viruses, bacteria, 
eukaryotic cells, and archaea. In~\cite{G97}, it was noted that, due to this, algorithms on 
circular strings may be important in the analysis of organisms with such structure. For instance, circular strings have been studied before 
in the context of sequence alignment. In~\cite{Lee:2010:FOA:1875737.1875765,circ09}, algorithms for multiple circular 
sequence alignment were presented. %These results were later improved in~\cite{}, where an additional preprocessing 
%stage was added to speed up the execution time of the algorithm. Later, in, the authors 
%presented efficient algorithms for finding the optimal circular consensus sequence and alignment for certain cases. 
Here we consider the problem of finding occurrences of a pattern $x$ of length $m$ with circular structure 
in a text $t$ of length $n$ with linear structure. 
%For instance, the DNA sequence of many viruses has circular structure, so if a biologist wishes to find occurrences 
%of a particular virus in a carriers DNA sequence---which 
%may not be circular---they must consider how to locate all positions in $t$ that at least one rotation of $x$ occurs. 
This is the problem of \emph{circular string matching}.

The problem of exact circular string matching has been considered in~\cite{Lot05}, where an $\cO(n)$-time algorithm was presented. 
The approach presented in~\cite{Lot05} consists of preprocessing $x$ by constructing 
a \emph{suffix automaton} of the string $xx$, by noting that every rotation of $x$ is a factor of $xx$. 
Then, by feeding $t$ into the automaton, the lengths of the longest factors of $xx$ occurring in $t$ can be found 
by the links followed in the automaton in time $\cO(n)$. 
In~\cite{Grabowski}, an average-case optimal algorithm for exact circular string matching was presented and it was also shown that the average-case lower bound for single string matching of $\Omega(n \log_\sigma m/m)$ also holds 
for circular string matching. Very recently, in~\cite{Chen03032013}, the authors presented two fast average-case algorithms 
based on word-level parallelism. The first algorithm requires average-case time $\cO(n \log_{\sigma}m/w)$, where $w$ is the 
number of bits in the computer word. The second one is based on a mixture of word-level parallelism and $q$-grams. 
The authors showed that with the addition of $q$-grams, and by setting $q = \Theta(\log_\sigma m)$, an average-case optimal time 
of $\cO(n \log_{\sigma}m/m)$ is achieved.  
Indexing circular patterns~\cite{Iliopoulos:2008:ICP:1787651.1787658} based on the construction of \emph{suffix tree}---have also been considered. 

The aforementioned algorithms for the exact case have the disadvantage that
they cannot be applied in a biological context since single nucleotide polymorphisms and errors introduced 
by wet-lab sequencing platforms might have occurred in the sequences; also it is not clear whether they could easily be adapted to deal with the approximate case. 
For the rest of the article, we assume that each position in the text $t$ is uniformly randomly drawn from $\Sigma$, 
and consider the following problem.

\defproblem{\ACSM}{a pattern $x$ of length $m$, a text $t$ of length $n>m$, and an integer threshold $k<m$}{all factors $u$ of $t$ such that $u \equiv_k^E x^i$, $0 \leq i < m$}

%In this problem we may only preprocess the pattern.
Similar to the exact case~\cite{Grabowski}, it can be shown that the average-case lower bound for single 
approximate string matching of $\Omega(n(k + \log_\sigma m) /m)$~\cite{Chang} also holds for approximate circular string matching 
under the edit distance model.
Recently, we have presented average-case $\cO(n)$-time algorithms for approximate circular string matching which are also very efficient in practice~\cite{1748-7188-9-9}. 
In~\cite{aproxcir}, an algorithm with $\cO(\frac{nk\log m}{m})$ average-case search time was presented.
To achieve average-case optimality, one could use the algorithms for multiple approximate string matching, 
presented in~\cite{Fredriksson:2004:ASM:1005813.1041513}, for matching the $r=m$ rotations of $x$ with
$\cO(n(k + \log_\sigma rm) /m)$ average-case search time, only if $k/m<1/2-\cO(1/\sqrt{\sigma})$ and $r= \cO(\min(n^{1/3} /m^2, \sigma^{o(m)}))$. 
Therefore the focus of this article is on a more {\em direct} algorithm which also improves on the preprocessing time and space complexity.

\textbf{Our Contribution.} 
In this article, we present a new average-case optimal algorithm for approximate circular string matching, under the edit distance model, that reduces 
the preprocessing time and space requirements compared to previous algorithms with optimal average-case search time. These savings are around $\cO(m^2)$ or more in all cases.
%We believe that the presented algorithm is more direct than existing algorithms for this problem
%as it does not rely on a reduction to multiple approximate string matching. 
%
\section{Algorithm}
\label{sec:algo}
In this section, we present our algorithm for approximate circular string matching under the edit distance model. 
The presented algorithm consists of two distinct schemes: the {\em searching} scheme, which determines if the currently considered text window potentially has a valid occurrence; in case the window \textit{may} contain a valid occurrence, we are required to check the window for valid occurrences of the pattern or any of its rotations; this is done through the {\em verification} scheme.

Intuitively, the algorithm considers a {\em sliding window} of length $m-k$ of the text, and reads $q$-grams backwards from the end of the window until it is likely to have found enough {\em differences} to skip the entire window. That is, we wish to make the probability of a verification being triggered sufficiently unlikely whilst also ensuring we can shift the window a reasonable amount.

The rest of this section is structured as follows. We first present an efficient incremental string comparison technique which forms the basis of the verification scheme. 
We then present the searching scheme of our algorithm which requires a preprocessing step. In fact, this preprocessing step is similar to the verification scheme. Finally, we show how plugging these schemes together results in a new average-case optimal algorithm for approximate circular string matching.

\subsection{Verification scheme}

The verification scheme of our algorithm is based on incremental string comparison techniques. 
First we give an introduction to these techniques; and then explain how we use them in the verification scheme.
The incremental string comparison problem was introduced in the pioneering work of Landau {\em et al}~\cite{LMS98}. 
The authors considered the following problem: given the edit distance between two strings $\textsf{A}$ and $\textsf{B}$, 
how can the edit distance between $\textsf{A}$ and $\texttt{b}\textsf{B}$ or $\textsf{B}\texttt{b}$ be efficiently derived, where $\texttt{b}$ is an additional letter. 
Given a threshold on the number of differences $k$, they solve this problem and allow prepending and appending of letters in time $\cO(k)$ per operation. 
Later the authors of~\cite{Hsu:2009:FAG:1696924.1697033} considered a generalisation of this problem with the aim of computing all maximal gapped palindromes in a string. 
The problem considered is a generalisation of the incremental string comparison problem considered in~\cite{LMS98} as it considers how to efficiently derive the edit distance when prefixes are deleted and letters are prepended to $\textsf{A}$ or $\textsf{B}$. The solution proposed in~\cite{Hsu:2009:FAG:1696924.1697033} also has a 
time complexity of $\cO(k)$ per operation. 
The solution for the generalised incremental string comparison problem forms the basis of our verification step. The technique lends itself more naturally to circular string matching due to the increased flexibility it provides. We begin by recalling some of the main results from~\cite{Hsu:2009:FAG:1696924.1697033} required for our algorithm.

The main idea in both~\cite{LMS98} and~\cite{Hsu:2009:FAG:1696924.1697033} is the efficient computation of the so-called \textit{$h$-waves}. 
In the standard dynamic programming matrix for two strings $x$ and $y$, we say that a cell $\textsf{D}[i,j]$ is on the diagonal $d$ {\em iff} $j-i=d$. For each diagonal, we may have a lowest cell with value $h$; if $\textsf{D}[i,j] = h$ and $\textsf{D}[i+1,j+1] = h+1$ then $\textsf{D}[i,j]$ is this cell for diagonal $j-i$. The $h$-wave, for all $0 \leq h \leq k$, is the position of all these cells across all diagonals, that is, a list $\textsf{H}_h$ of length $\cO(k)$, where each entry is a pair $(i,j)$ such that $\textsf{D}[i,j] = h$ and $\textsf{D}[i+1,j+1] = h+1$. Note that the $i$-th wave can only contain entries on diagonal zero and the $i$ diagonals either side of it, so for $0 \leq i \leq k$ every wave has size $\cO(k)$. 
Both incremental string comparison techniques show some bounds on the possible values of the cells on $h$-waves and how to 
efficiently compute them. 
These $h$-waves define the entire dynamic programming matrix due to the monotonicity properties of the matrix. For any diagonal $d$, if we know the position of the lowest cell on $d$ with 
value $h$ and $h+1$, then we also know the value of every cell between these two cells: it must be $h+1$. So given the $h$-waves of the matrix, for all $0 \leq h \leq k$, we have all the information that is in the standard dynamic programming matrix. The key result from our perspective is the following.

Let $\textsf{cat}(u', u)$ denote the string obtained by 
concatenating $u'$ and $u$, where $u,u' \in \Sigma^+$.
Let $\textsf{del}(\alpha, u)$ denote the string obtained by deleting the 
prefix of length $\alpha$ of $u$.
Let $\textsf{D}^{\prime}$ denote the standard dynamic 
programming matrix for strings $\textsf{cat}(\textsf{A}^{\prime}, \textsf{A})$ and 
$\textsf{del}(t_2,\textsf{B})$, where $|\textsf{A}^{\prime}|=t_1$.

\begin{theorem}[\cite{Hsu:2009:FAG:1696924.1697033}]
The $0$-wave, $1$-wave, $\ldots$ , and $k$-wave of matrix $\textsf{D}^{\prime}$ can be computed in time $\cO((t_1 
+ t_2)k)$.
\label{the:Hsu}
\end{theorem}

If a window of the text triggers a verification then we have a window of length $m-k$ such that there exist some $q$-grams of the window that occur in $x$ or its rotations with at most $k$ differences in total. When we verify a window, we check for occurrences of pattern $x$ starting at every position in the window. For each position, we may have a factor of length at most $m+k$ representing an occurrence, meaning we must consider a factor $w$ of the text of length $2m$ which we refer to as a \textit{block}. This ensures we avoid missing any occurrences at the $m-k$ starting positions as $(m-k) + (m+k)=2m$.

For each possible starting position $i$, $0 \leq i < m-k$, we compute the $0$-wave, $1$-wave, $\dots$ , and $k$-wave for $x$ 
and $w'=w[i \dd 2m-1]$, the suffix of $w$ starting at position $i$. To check if we have an occurrence, we must check the $k$-wave $\textsf{H}_k$. 
We iterate through each entry in the $k$-wave $\textsf{H}_k$; and if $\textsf{H}_k$ has missing entries or contains entries on the last row of the matrix, then $x$ occurs in $w$ with at most $k$ differences. 

Similarly we can check for the occurrences of the rotations of $x$ using the incremental string comparison techniques.
We are now ready to outline the verification scheme, denoted by function $\textsf{VER}$. Given the pattern $x$ of length $m$, an integer threshold $k<m$, and a block $w$
of length $2m$ of the text $t$, function $\textsf{VER}$ finds all factors $u$ of $w$ such that $u \equiv_k^E x^i$, $0 \leq i < m$.
If any diagonal has no entry on the $k$-wave then that diagonal reached the last row of the matrix with less than $k$ differences; this means $x$ occurs in $w$ with less than $k$ differences.
  \begin{myfunction}[H]{10 cm}
  \Algo{$\textsf{VER}(x,m,k,w,2m)$}{
      Compute the edit distance between $x$ and $w'=w[0\dd 2m-1]$ with at most $k$ differences using the standard dynamic programming algorithm\;
      Check for any occurrences using $\textsf{D}$, and if found, {\bf report} an occurrence at position 0\;
      \ForEach{$i\in \{1,m-k-1\}$}{
       \ForEach{$j\in \{1,m\}$}{
       Construct rotation $x^{j}$ of $x$ by removing the first letter of $x^{j-1}$ and appending it to the end of $x^{j-1}$\; 
       Compute the edit distance between $x^{j}$ and $w'=w[i\dd 2m-1]$ using the incremental string comparison techniques\;
       Check for any occurrences using $\textsf{H}_k$, and if found, {\bf report} an occurrence at the current position $i$ being checked\;
       }
     }
  }
  \end{myfunction}
\begin{lemma}
Given the pattern $x$ of length $m$, an integer threshold $k<m$, and string $w$ of length $2m$, function $\textsf{VER}$ requires time $\cO(m^2k)$.
\label{lem:ver}
\end{lemma}
\begin{proof}
Computing the edit distance between $x$ and $w[0\dd 2m-1]$ with at most $k$ differences takes time $\cO(mk)$ using the standard dynamic programming algorithm.
By Theorem~\ref{the:Hsu}, computing the edit distance between all the rotations of the pattern and $w[i\dd 2m-1]$ for a single position in $w$ requires $\cO(mk)$; 
and there are $\cO(m)$ positions in $w$. In total, the time is $\cO(mk + m^2k)$, that is $\cO(m^2k)$. \qed
\end{proof}
\subsection{Searching scheme}
The searching scheme of the presented algorithm requires the preprocessing and indexing of the pattern $x$. 
We first present the preprocessing required and then present the searching technique itself.
\subsubsection*{Preprocessing.}
We build a $q$-gram index in a similar way as that proposed by Chang and Marr in~\cite{Chang}. Intuitively, we wish to determine the minimum possible edit distance between every $q$-gram and any factor of $x$ or its rotations. Equivalently we find the minimum possible edit distance between every $q$-gram and any {\em prefix} of a factor of length $2q$ of $x$ and the suffixes of length 1 to $2q$ of $x$ or its rotations. An index like this allows us to lower bound the edit distance between a window of the text and $x$ or its rotations without computing the edit distance between them. To build this index, we generate every string of length $q$ on $\Sigma$, and find the minimum edit distance between it and all prefixes of factors of length $2q$ of $x$ or its rotations. This information can easily be stored by generating a numerical representation of the $q$-gram and storing the minimum edit distance in an array at this location. If we know the numerical representation, we can then look up any entry in constant 
time. 

We determine the edit distance using the preprocessing scheme, denoted by function $\textsf{PRE}$, which is 
similar to the verification scheme (function $\textsf{VER}$). 

Given the string $x'=x[0\dd m-1]x[0 \dd m-2]$ of length $2m-1$, function $\textsf{PRE}$ finds the minimum edit distance between every $q$-gram on $\Sigma$, generated in increasing order, and any factor $u$ of length $2q$ of $x'$ and its suffixes of length 1 to $2q$.

\begin{lemma}\label{lemm:pre}
Given the string $x'=x[0\dd m-1]x[0 \dd m-2]$ of length $2m-1$ on $\Sigma$, $\sigma = |\Sigma|$, and $q < m$, 
function  $\textsf{PRE}$ requires time $\cO(\sigma^qmq)$ and space $\cO(\sigma^q)$.
\end{lemma}

\begin{proof}
The time required for initialising array $\textsf{M}$ is $\cO(\sigma^q)$.
The time required for computing the edit distance between $x'[0\dd 2q-1]$ and $s$ is $\cO(q^2)$ using the standard dynamic programming algorithm. 
By Theorem~\ref{the:Hsu}, computing the edit distance between all $2q$-grams of $x'$ and $s$ requires time $\cO(mq)$.
There exist $\cO(\sigma^q)$ possible $q$-grams on $\Sigma$ and so, in total, the time complexity is $\cO(\sigma^qmq)$. Keeping array $\textsf{M}$ in memory requires space $\cO(\sigma^q)$. \qed
\end{proof}

  \begin{myfunction}[H]{11 cm}
  \Algo{$\textsf{PRE}(x',2m-1,q,\sigma)$}{
      $\textsf{M}[0\dd \sigma^q -1]\leftarrow 0$\;
      $j \leftarrow 0$\;
      \ForEach{$s \in \Sigma^{q}$}{
      %\textit{Where $s$ is considered in lexicographical order}
      Compute the edit distance between $u=x'[0\dd 2q-1]$ and $s$ using the standard dynamic programming algorithm. Set $E_{\min}$ equal to the minimum edit distance between $s$ and any prefix of $u$ using $\textsf{D}$\;
       \ForEach{$i\in \{1,2m-q-1\}$}{
       $u \leftarrow x'[i \dd \textsf{min} \left\lbrace i + 2q-1, 2m-2 \right\rbrace ]$\;
       Compute the edit distance $E'$ between $u$ and $s$ using the incremental string comparison techniques. Set $E'$ equal to the minimum edit distance between $s$ and any prefix of $u$ using $\textsf{H}_{q}$\;
       \lIf{$E' < E_{\min}$}{$E_{\min} \leftarrow E'$}
     }
       $\textsf{M}[j]\leftarrow E_{\min}$\;
       $j \leftarrow j+1$\;
    }
    \Return{$\textsf{M}$}\;
  }
  \end{myfunction}

\subsubsection*{Searching.}
In the search phase we wish to read backwards enough $q$-grams from a window of size $m$ that the probability we must verify the window is small and the amount we can shift the window by is sufficiently large.
We now recall some important lemmas from~\cite{Chang} that we will use in the analysis of our algorithm.

\begin{lemma}[\cite{Chang}]\label{lem:prob}
The probability that two $q$-grams on $\Sigma$, one being uniformly random, have a common subsequence of length $(1-c)q$ is 
at most $\frac{a\sigma^{-dq}}{q}$, where $a = (1+o(1))/(2\pi c(1-c))$ and $d = 1 - c + 2c \log_\sigma c + 2(1 - c) \log_\sigma (1 - c)$. The probability decreases exponentially for $d>0$, which holds if $c < 1- \frac{e}{\sqrt{\sigma}}$.
\label{lem:Chang1}
\end{lemma}
\begin{lemma}[\cite{Chang}]
If $s$ is a $q$-gram occurring with less than $cq$ differences in a given string $u$, $|u|\geq  q$, $s$ has a common subsequence of length $q-cq$ with some $q$-gram of $u$.
\label{lem:Chang2}
\end{lemma}
By Lemmas~\ref{lem:Chang1} and~\ref{lem:Chang2}, we know that the probability of a random $q$-gram occurring in a string of length $m$ with less than $cq$ differences is no more than $ma\sigma^{-dq}/q$ as we have $m-q+1$ $q$-grams in the string. For circular string matching this is not sufficient. To ensure that we have the $q$-grams of all possible rotations of pattern $x$, we instead consider the string $x'=x[0\dd m-1]x[0 \dd m-2]$ and extract the $q$-grams from $x^{\prime}$. We may have up to $2m-q$ $q$-grams, but to simplify the analysis we assume we have $2m$ and so the probability becomes $2ma\sigma^{-dq}/q$. 

In the case when we read $k/(cq)$ $q$-grams, we know that with probability at most $(k/(cq))2ma\sigma^{-dq}/q$ we have found less than $k$ differences. This does not permit us to discard the window if all $q$-grams occur with at most $cq$ differences. To fix this, we instead read $1 + k/(cq)$ $q$-grams. If any $q$-gram occurs with less than $cq$ differences, we will need to verify the window; but if they all occur with at least $cq$ differences, we must exceed the threshold $k$ and can shift the window. When shifting the window we have the case that we shift after verifying the window and the case that the differences exceed $k$ so we do not verify the window. If we have verified the window, we can shift past the last position we checked for an occurrence: we can shift by $m-k$ positions. If we have not verified the window, as we read a fixed number of $q$-grams, we know the minimum-length shift we can make is one position past this point. The length of this shift is at least $m-k-(q+k/c)$ positions. This 
means 
we will have at most $\frac{n}{m-k-(q+k/c)}= \cO(\frac{n}{m})$ windows. The previous statement is only true assuming $m-q>k+k/c$, as then the denominator is positive. From there we see that we also have the condition that $q + k + k/c$ can be at most $\epsilon m$, 
where $\epsilon < 1$, so the denominator will be $\cO(m)$. This puts a slightly stricter condition on $c$, that is, $c> \frac{k}{\epsilon m -q -k}$.

We can see that, for each window, we verify with probability at most $(1 + k/(cq))2ma\sigma^{-dq}/q$, where $a = (1+o(1))/(2\pi c(1-c))$ 
and $d = 1 - c + 2c \log_\sigma c + 2(1 - c) \log_\sigma (1 - c)$. So the probability that a verification is triggered is

$$\frac{(1+k/(cq))2ma\sigma^{-dq}}{q}.$$

\noindent Because by Lemma~\ref{lem:ver}, verification takes time $\cO(m^2k)$, then per window, the expected cost is

$$\frac{(1+k/(cq))2ma\sigma^{-dq}\cO(m^2k)}{q} = \cO(\frac{(q+k)m^3ka\sigma^{-dq}}{q^2}).$$

\noindent We wish to ensure that the probability of verifying a window is small enough that the average work done is no more than the work we must do if we skip a window without verification. When we do not verify a window, we read $1 + k/(cq)$ $q$-grams and shift the window. This means that we read $q + k/c = \cO(q + k)$ letters. So a sufficient {\em condition} is the following:

$$\frac{(q+k)m^3ka\sigma^{-dq}}{q^2} = \cO(q + k).$$

\noindent Or equivalently the below expression, where $f$ is the constant of proportionality:

$$\frac{(q+k)m^3ka\sigma^{-dq}}{q^2} \leq f(q + k).$$

\noindent By rearranging and setting $f = \sigma$ we get the condition on the value of $q$ below:

$$q \geq \frac{3\log_\sigma m + \log_\sigma k + \log_\sigma a - 2\log_\sigma q}{d}.$$

\noindent From the condition on $q$ we can see that it is sufficient to pick $q = \Theta(\log_\sigma km)$, so asymptotically on $m$ we get the following:

$$q \geq \frac{3\log_\sigma m + \log_\sigma k - \cO(\log_\sigma \log_\sigma km)}{d}.$$

\noindent Therefore, for sufficiently large $m$, the below condition is sufficient for optimality, where $d=1 - c + 2c \log_\sigma c + 2(1 - c) \log_\sigma (1 - c)$:

$$q = \frac{3 \log_\sigma m + \log_\sigma k}{d}.$$

\noindent For this analysis to hold we must be able to read the required number of $q$-grams to ensure the probability of verifying a window is small enough to negate the work of doing it. Note that the above probability is the probability that at least one of $q$-grams match with less than $cq$ differences. To ensure we have enough unread random $q$-grams in the window for Lemma~\ref{lem:Chang2} to hold in the above analysis the window must be of size $m-k \geq 2q + 2k/c$. Now we consider the case where $2q + 2k/c > m-k \geq 2q + k/c$. If we have just verified a window then we have enough new random $q$-grams and our analysis holds. If we have just shifted then we know that all the $q$-grams we previously read matched with at least $cq$ differences and we have between 1 and $k/qc$ $q$-grams and the probability that one of these matches with less than $cq$ difference is less than in the analysis above so it holds.

The condition $m-k \geq 2q + k/c$ implies a condition on $c$, it must be the case that $c \geq \frac{k}{m-k-2q}$. This condition on $c$ is weaker than our previous condition on $c$, so to determine the {\em error ratio} $\frac{k}{m}$, we use the stronger condition. Additionally, from Lemma \ref{lem:prob}, we know that $c < 1 - \frac{e}{\sqrt{\sigma}}$.  So we must pick a value for $c$ subject to $ \frac{k}{\epsilon m-k-q} \leq c < 1- \frac{e}{\sqrt{\sigma}}$.
This inequality implies a limit on the error ratio for which our algorithm is optimal. Clearly it must be the case that $\frac{k}{\epsilon m-k-q} <1- \frac{e}{\sqrt{\sigma}}$ for $\epsilon <1$. 
Rearranging the inequality implies the following sufficient condition on our error ratio: 

$$\frac{2k}{m} < \epsilon - \frac{q}{m} - \frac{\epsilon e}{\sqrt{\sigma}} + \frac{qe}{m \sqrt{\sigma}} + \frac{ke}{m\sqrt{\sigma}}.$$

\noindent From here we can factorise and divide everything by 2 to get the following:

$$\frac{k}{m} < \frac{\epsilon}{2} - \frac{q}{2m} - \frac{e}{2\sqrt{\sigma}}(\epsilon - \frac{q}{m} - \frac{k}{m}).$$

\noindent So asymptotically on $m$ we have:

$$ \frac{k}{m} < \frac{\epsilon}{2}- \cO(\frac{1}{\sqrt{\sigma}}).$$

\noindent Note that this technique can work for any ratio which satisfies $\frac{k}{m} < \frac{1}{2}- \cO(\frac{1}{\sqrt{\sigma}})$.
For any ratio below this, pick a large enough value for $\epsilon$ such that asymptotically on $m$ the algorithm will work in the claimed search time.
By choosing a suitable value for $c$ and $q \geq \frac{3 \log_\sigma m + \log_\sigma k}{d}$ we obtain the following result.
%
%{\bf Because for any ratio less than the one above we can pick always pick an $\epsilon$ such that $\epsilon - 1/2$ is bigger than that error ratio, or even pick $\epsilon$ arbitrarily close to 1 and then pick an $m$ large enough that the q/m factors which are hidden in the $\cO$ can overcome any constant difference between them as they become arbitrarily close to 0}

\begin{theorem}
The problem \ACSM\ can be solved in optimal average-case search time $\cO(n(k+ \log_\sigma m)/m)$.
\end{theorem}

\section{Comparison with Existing Algorithms}

To the best of our knowledge, the only other algorithms to achieve optimal average-case search time for approximate circular string matching are the algorithms presented in~\cite{Fredriksson:2004:ASM:1005813.1041513} for multiple approximate string matching. In the analysis of the algorithms in~\cite{Fredriksson:2004:ASM:1005813.1041513} it is assumed that all patterns are random. In~\cite{doi:10.1142/S0129054106004455} the authors re-analyse their algorithms for the problem of circular string matching with the same preprocessing and space costs. 
In this section, we analyse these results and compare them with our own. 
We refer to the algorithm presented in Section~\ref{sec:algo} as $\textsf{BIP}$. Due to the constant $c$ in the value of $q$ from Lemma \ref{lem:Chang1}, the exact preprocessing and space costs for these algorithms depend on the chosen value for $c$. It is however possible to determine the minimum savings we make based on the value of $q$ used in all algorithms.

Applying the algorithms in~\cite{Fredriksson:2004:ASM:1005813.1041513} to approximate circular string matching requires a reduction to multiple approximate string matching for matching the $m$ rotations of $x$. The first algorithm in~\cite{Fredriksson:2004:ASM:1005813.1041513} has the following time complexity:

$$\cO(n(k + \log_\sigma rm) /m).$$

\noindent By setting $r=m$ this matches our search time and the result is valid when $k/m<1/2-\cO(1/\sqrt{\sigma})$, $r= \cO(\min(n^{1/3} /m^2, \sigma^{o(m)}))$, and we have $\cO(\sigma^q)$ space available, where $q$ is subject to the constraint: 

$$q \geq \frac{4 \log_\sigma m + 2 \log_\sigma r}{d}.$$

\noindent Again by setting $r=m$ this becomes
$q \geq \frac{6 \log_\sigma m}{d}$ and the preprocessing time is $\cO(\sigma^q m^2)$. We will refer to this algorithm as $\textsf{FN1}$. The second algorithm, presented in~\cite{Fredriksson:2004:ASM:1005813.1041513}, has the same preprocessing cost and requires space $\cO(\sigma^q m)$. We will refer to this algorithm as $\textsf{FN2}$. The important difference between \textsf{FN1} and \textsf{FN2} comes in the condition on $q$ which is slightly lower for \textsf{FN2}:

$$q \geq \frac{3 \log_\sigma m + \log_\sigma r + \log_\sigma (m + \log_2 r)}{d}.$$

\noindent Again, setting $r=m$ this becomes:

$$q \geq \frac{4 \log_\sigma m + \log_\sigma (m + \log_2 m)}{d}.$$

To simplify the comparison between these approaches, we will ignore the factor of $\log_2 m$, and simply say that the value of $q$ for algorithm $\textsf{FN2}$ is greater than or equal to $\frac{5 \log_\sigma m}{d}$. This is lower than the sufficient requirement, so any saving we make using this value must be at least as good or better in reality.

First let us consider \textsf{FN1}. The preprocessing requirement of $\textsf{BIP}$ is $\cO(\sigma^qmq)$, so before any savings made due to the value of $q$ for $\textsf{BIP}$, we have reduced the preprocessing cost by a factor of $\cO(\frac{m}{q})$. Given the condition on $q$ for $\textsf{BIP}$, it is clear that even in the worst case, when $k=\cO(m)$, $\textsf{BIP}$ will make a saving of at least $2 \log_{\sigma} m$ on the value of $q$. This corresponds to an additional saving of $\cO(m^2)$ in preprocessing time bringing the total to $\cO(\frac{m^3}{q})$ and $\cO(m^2)$ in space. 
In the case of $\textsf{FN2}$, we make a saving of at least $\log_{\sigma} m$ on the value of $q$. This corresponds to a total saving of $\cO(\frac{m^2}{q})$ in preprocessing time and $\cO(m^2)$ in space. It should be noted that this is a pessimistic analysis of the savings as we have assumed $k=\cO(m)$ and $d=1$, although it must hold that $d<1$. Note that the standard dynamic programming algorithm can be used with runtime $\cO(m^3)$ for verification and $\cO(\sigma^q m q^2)$ for preprocessing. The speed-ups mentioned in the previous section remain significant as we assumed
that $k=\cO(m)$. We still achieve a preprocessing speed up of at least $\cO(m^2)$ and $\cO(m)$ against \textsf{FN1} and \textsf{FN2}, respectively.
Table~\ref{tab:comparison} corresponds to this analysis. 

\begin{table}[h]
\begin{center}
\caption{Comparison of average-case optimal approximate circular string matching algorithms}
\label{tab:comparison}
\begin{tabular}{|l|l|l|l|l|}
\hline Algorithm & Error Ratio ($k/m$) & Space & Preprocessing Time & Condition on $q$ \\ \hline
\textsf{FN1} & $\frac{1}{2}- \cO(\frac{1}{\sqrt{\sigma}})$ & $\cO(\sigma^q)$ & $\cO(\sigma^q m^2)$ & $\frac{6 \log_\sigma m}{d}$\\ \hline
\textsf{FN2} & $\frac{1}{2}- \cO(\frac{1}{\sqrt{\sigma}})$ & $\cO(\sigma^qm)$ & $\cO(\sigma^qm^2)$ &  $\frac{4 \log_\sigma m + \log_\sigma (m + \log_2 m)}{d}$\\ \hline 
\textsf{BIP} & $\frac{1}{2}- \cO(\frac{1}{\sqrt{\sigma}})$ & $\cO(\sigma^q)$ & $\cO(\sigma^q mq)$ & $\frac{3 \log_\sigma m + \log_\sigma k}{d}$\\ \hline
\end{tabular}
\end{center}
\end{table}

\section{Final Remarks}

In this article, we presented a new average-case optimal algorithm for approximate circular string matching. 
To the best of our knowledge, this algorithm is the first average-case optimal algorithm specifically designed for this problem. 
Other average-case optimal algorithms exist but with higher preprocessing and space requirements than the presented algorithm. 
Additionally the considered problem is solved in a more direct fashion, that is, with no reduction to multiple approximate string matching 
by taking greater advantage of the similarity of the rotations of the pattern. 

\noindent Our immediate target is twofold: 
\begin{itemize}
\item first, we plan on tackling the problem of multiple approximate circular string matching. We will try to generalise the approach we have taken here to see if it leads to an average-case optimal 
algorithm in this case. 
\item second, we plan on implementing the presented algorithm. We will then compare the respective implementation to other average- and worst-case approaches.
\end{itemize}

\bibliographystyle{splncs03}
\bibliography{references}
\end{document}